\documentclass[%
 reprint,
 amsmath,amssymb,
 aps,
]{revtex4-2}

\usepackage{graphicx}
\usepackage{dcolumn}
\usepackage{bm}
\usepackage{ulem}
\usepackage{xcolor}
\usepackage{amsmath}
\usepackage{amsthm}

\def\sideremark#1{\ifvmode\leavevmode\fi\vadjust{\vbox to0pt{\vss
 \hbox to 0pt{\hskip\hsize\hskip1em
 \vbox{\hsize2cm\tiny\raggedright\pretolerance10000
  \noindent #1\hfill}\hss}\vbox to8pt{\vfil}\vss}}}

\newcommand{\II}{{\rm  I\hspace{-.2mm}I}}

\newcommand\eqSig{\mathrel{\stackrel{\mathcal{H}\hh}{=}} }
\newcommand{\hh}{{\hspace{.3mm}}}
\newcommand{\Span}[1]{\langle #1 \rangle}
\newcommand{\ext}{{\rm{d}}}

\newtheorem{theorem}{Theorem}
\newtheorem{definition}{Definition}
\newtheorem{proposition}{Proposition}
\newtheorem{corollary}{Corollary}
\theoremstyle{remark}
\newtheorem*{remark}{Remark}
\newtheorem*{slogan}{Slogan}

\def\beq#1\eeq{\begin{align}#1\end{align}}

\begin{document}

\preprint{APS/123-QED}


\title{Horizons that Gyre and Gimble:\\ A Differential Characterization of Null Hypersurfaces}

\author{Samuel Blitz}
\affiliation{
 Department of Mathematics and Statistics, Masaryk University\\ Building 08, Kotlářská 2,
Brno, CZ 61137 \\
blitz@math.muni.cz
 }%

\author{David McNutt}
\affiliation{ Department of Mathematics and Statistics, UiT: The Arctic University of Norway, Tromsø, Norway, 9019 \\ 
david.d.mcnutt@uit.no
}%

\date{\today}

\begin{abstract}

Motivated by the thermodynamics of black hole solutions conformal to stationary solutions, we study the geometric invariant theory of null hypersurfaces. 
It is well-known that a null hypersurface in a Lorentzian manifold can be treated as a Carrollian geometry. Additional structure can be added to this geometry by choosing a connection which yields a Carrollian manifold. In the literature various authors have introduced Koszul connections to study the study the physics on these hypersurfaces. In this paper we examine the various Carrollian geometries and their relationship to null hypersurface embeddings. We specify the geometric data required to construct a rigid Carrollian geometry, and we argue that a connection with torsion is the most natural object to study Carrollian manifolds. We then use this connection to develop a hypersurface calculus suitable for a study of intrinsic and extrinsic differential invariants on embedded null hypersurfaces; motivating examples are given, including geometric invariants preserved under conformal transformations.

\end{abstract}

\maketitle


\section{\label{sec:intro}Introduction\protect}

\bigskip

Black holes have captured the imagination of both the relativist and the lay public for decades. In the last 10 years, observational data of in-spiraling black holes  from the LIGO and Virgo collaborations~\cite{GWobservation} have sparked greater interest in dynamical black holes with complicated event horizon structures. A particularly useful formalism in a relativists' toolbox in this respect is black hole thermodynamics. To describe black holes with respect to thermodynamic variables, we must locate the event horizon or some other meaningful hypersurface which acts as the boundary for the black hole, and then characterize it geometrically. For example, the black hole entropy is associated with the area of the horizon.

When considering the simplest cases (such as static and stationary solutions), there is broad agreement that the necessary horizon is a Killing horizon of some timelike Killing vector~\cite{wald2010general}, and there is a wealth of methods to discuss black hole thermodynamics in this regime. In the dynamical case, the question of which hypersurface is the appropriate boundary and what its thermodynamical properties might be is an open problem \cite{booth2005black}. While some of the candidate hypersurfaces, such as dynamical horizons, are spacelike \cite{booth2005black}, there are convincing arguments that this hypersurface should be null \cite{wall2012proof}.

In particular, there is a special class of dynamical black hole solutions which are conformally related to stationary black hole solutions, and hence admit a conformal Killing vector field. Within this class of black hole solutions, there is an obvious candidate for a horizon alternative, namely the conformal Killing horizon which is a null hypersurface \cite{sultana2004conformal}. Due to the properties of a conformal mapping, there is a procedure to relate the properties of stationary spacetimes to such dynamical spacetimes. For these black hole solutions, it then should be possible to determine the thermodynamic properties of the stationary solution and relate them to the dynamical solution. From a physical perspective, it is expected that the outcomes of experiments in both solutions are equivalently mapped to each other so long as the effects of the conformal transformation on the coupling of the geometrical spacetime and the matter degrees of freedom within the experimental apparatus are taken into account \cite{dicke1962mach, faraoni2007pseudo, codello2013renormalization}.

However, there is a difficulty with this proposal: it is not clear when a conformal transformation of a black hole solution yields a new black hole solution. Taking the simplest static solution, Schwarzschild, and a time dependent conformal factor, it is possible to construct new dynamical black hole solutions, cosmological solutions, or more exotic spacetimes \cite{mello2017}. Alternatively, with a spatial conformal factor, it is possible to transform any spherically symmetric black hole solution into a Kundt solution, which will not describe a black hole solution \cite{pravda2017exact}. Due to this nuance in the choice of conformal factor, and motivated by the utility of conformal Killing horizons in the study of dynamical black hole solutions, we will classify the intrinsic and extrinsic invariants of embedded null hypersurfaces with the future aim of employing these results to characterize dynamical black hole solutions admitting conformal Killing horizons.

Introduced in 1965 by Levy-Leblond~\cite{LvyLeblond1965}, Carrollian geometries have recently~\cite{Penna2018,Donnay_2019,Ciambelli2019,Chandrasekaran2022,Petkou2022} been investigated as an avenue of describing and characterizing event horizons of many kinds and null hypersurfaces (such as null infinities~\cite{Duval2014,Hartong2015,Herfray2020,Herfray2022,Prabhu2022,Bagchi2022,Liu2023}), more generally. Carrollian geometries have also been studied in other, more intrinsic, settings, such as ultra-relativistic fluid dynamics~\cite{Ciambelli2018,Ciambelli2018-2} and Carrollian field theories~\cite{duval2014carroll, Bagchi2019,Bagchi2020,Bagchi2021,Banerjee2021,baiguera2023conformal,rivera2022revisiting,henneaux2021carroll}. In this article, we offer a geometric picture of Carrollian manifolds that harmonizes the interesting physics intrinsic to a Carrollian manifold
with the extrinsic geometric invariants necessary to classify a certain family of null hypersurface embeddings, which largely intersects with conformal-to-stationary black hole spacetimes.

\subsection{Notation}
Here we provide a brief summary of the notations that we will use throughout this article. The manifold $M$ represents a $d$-dimensional spacetime which is endowed with a (mostly negative) Lorentzian signature metric $g$.

Given such a metric, its associated Levi-Civita connection $\nabla$ has curvature given by
$$R(x,y)z = (\nabla_x \nabla_y - \nabla_y \nabla_x)z - \nabla_{[x,y]} z\,,$$
where $x,y,z$ are smooth vector fields on $M$ and $[\cdot,\cdot]$ denotes the Lie bracket of vector fields. When coordinates are useful, we will use  indices from the Greek alphabet (such as $\alpha$, $\beta$, and $\gamma$ for spacetime coordinate indices).

To keep track of tensor structure, we will also often use Penrose's abstract index notation~\cite{Penrose1968}. These abstract indices will be denoted by Latin letters at the beginning of the alphabet (such as $a$, $b$, $c$, etc.). As an example, the above formula for the Riemann curvature tensor is given, in abstract indices, by
$$x^a y^b R_{ab}{}^c{}_d z^d = x^a y^b (\nabla_a \nabla_b - \nabla_b \nabla_a)z^c\,.$$
This abstract index notation will also be used to keep track of symmetry properties of a tensor: for example, we will use round brackets $(\cdot)$ to represent the symmetrization of a tensor, i.e. $T_{(ab)} = \tfrac{1}{2}(T_{ab} + T_{ba})$. We denote by the symbol $\odot$ the operator that maps a tensor to its symmetric part. Similarly, square brackets $[\cdot]$ will be used to represent the antisymmetrization of a tensor, so that $T_{[ab]} = \tfrac{1}{2} (T_{ab} - T_{ba})$. When a metric $g$ is present, we will occasionally use the notation $\langle u,v\rangle_{g}$ to mean $g(v,w)$, for some vectors $v,w$.

An embedded null hypersurface, denoted by $\mathcal{H}$, will always be assumed to be smooth, with dimension $n=d-1$. By abuse of notation, we will use $\mathcal{H}$ to refer both to an embedded null hypersurface in $M$ and as an $n$-dimensional manifold in its own right. When necessary to avoid overloading of symbols, we will use overbars $\bar{\bullet}$ to denote objects that belong to the hypersurface. For example, for a 1-form $\alpha$ on $M$, we might write $\bar{\alpha}$ for its pullback to $\mathcal{H}$. To eliminate confusion that might arise when considering null frames, we will use Latin letters to denote basis vectors, and their corresponding Greek letters to denote the action of the metric on those basis vectors, i.e. $\nu_a = g_{ab} n^b$.

When necessary, coordinate indices in $M$ will use letters from the beginning of the Greek alphabet and coordinate indices in $\mathcal{H}$ will use letters from the middle of the Greek alphabet, such as $\mu$, $\nu$, and $\rho$. We will label spacelike coordinates with capital letters from the beginning of the Latin alphabet, such as $A$, $B$, and $C$. Furthermore, a frame will sometimes be constructed on $\mathcal{H}$, and the indices used to label the spacelike vectors or covectors will be letters from the middle of the Latin alphabet, like $i$, $j$, and $k$. The indices for the whole frame on $\mathcal{H}$ will be capital letters from the middle of the Latin alphabet, like $I$, $J$, and $K$.

 
\section{Differential Structure of Carrollian Geometries}

In the literature broadly, Carrollian geometries generically refer to smooth manifolds with a minimally degenerate metric. Depending on the source, they might also include additional geometric structures such as a distinguished generating vector field, a rigging vector~\cite{mars1993} (or equivalently an Ehresmann connection~\cite{freidel2022}), and/or an affine connection~\cite{Ciambelli2018-2,Ciambelli2019,Chandrasekaran2022}. To be explicit, we provide a definition of the least rigid such structure below, drawing from definitions provided by~\cite{penrose1972} and~\cite{bekaertmorand2018}.
\begin{definition}
    Let $\mathcal{H}^n$ be a smooth manifold equipped with a rank-$(n-1)$ negative semi-definite symmetric bilinear form $\bar{g}$. Then the doublet $(\mathcal{H},\bar{g})$ is called a \textit{pre-Carrollian structure}. When $\iota : \mathcal{H}^n \hookrightarrow (M,g)$ is a smooth null hypersurface embedding into a Lorentzian spacetime and $\bar{g} = \iota^* g$, we say that $(\mathcal{H},\bar{g})$ is a \textit{null hypersurface structure} (NHS).
\end{definition}
The latter definition of an NHS was introduced in \cite{nurowski2000}. It is important to note that not all pre-Carrollian structures are null hypersurface structures; however, they can always be identified locally.  As we are only interested in local geometry, for the remainder of this paper, we will implicitly assume all pre-Carrollian structures can be realized as NHSs. In the context of black hole solutions, pre-Carrollian structures arise as non-expanding horizons (NEHs), which are a precursor to weakly isolated horizons \cite{ashtekar2002}. These are null hypersurfaces where there is no preferred choice of null generator and the pullback of the spacetime metric gives a degenerate metric on the NEH.

While these hypersurface embeddings describe very generic spacetimes, often spacetimes of physical interest have more geometric structure that can be exploited. To that end, we are motivated to add additional geometric data to our discussion of pre-Carrollian structures. Indeed, by specifying a vector that spans the radical of $\bar{g}$, we obtain~\cite{bekaertmorand2018}:
\begin{definition}
    Let $(\mathcal{H}^n,\bar{g})$ be a pre-Carrollian structure and let a vector field $\bar{\ell} \in \Gamma(T \mathcal{H})$ span the radical of $\bar{g}$, so that $\bar{g}(\bar{\ell},\cdot) = 0$. Then the triplet $(\mathcal{H},\bar{g},\bar{\ell})$ is  called a \textit{Carrollian structure}, and the vector field $\bar{\ell}$ is called its \textit{fundamental vector field}.
\end{definition}

\noindent Carrollian structures are also known as weak Carrollian structures in the literature \cite{duval2014conformal} and appear in the context of Carrollian field theories in the hydrodynamic regime \cite{armas2023carrollian,de2022carroll}. Any black hole solution admitting a weakly isolated horizon may be considered as a weak Carrollian geometry.

While for a given non-expanding horizon, one can assign any choice of $\bar{\ell}$ as the fundamental vector field, there is, in general, no natural choice. However, as noted above, for weakly isolated horizons~\cite{ashtekar2002} (and, in particular, Killing horizons), there exists a canonical choice of $\bar{\ell}$ for $\mathcal{H}$ (up to constant rescaling, which plays no role in this article). Thus, Carrollian structures are the natural geometric structures to examine when studying such spacetimes.

\medskip

One of the most efficient ways of describing the intrinsic differential structure of a smooth manifold is to assign that manifold a geometrically determined linear affine connection (sometimes called a \textit{Koszul connection}) so that one may study tensorial quantities built from curvatures. For a given Carrollian structure, there are infinitely many such connections that one may assign; however, the space of connections should not be viewed as arbitrary if one still wishes to respect the geometric data provided by the Carrollian structure. Broadly, this notion is captured by viewing manifolds with geometric data as G-structures over that manifold. To illustrate how we may pick such a family of connections, we consider an example pulled from~\cite{figueroa2020}: the Lorentzian spacetime $(M^d,g)$. 

Such a Lorentzian spacetime can be viewed as an $O(d-1,1)$-structure over $M$. One of the properties of a G-structure is that it gives rise to one or more characteristic tensor fields that are left invariant by the structure group. In this case, that characteristic tensor field is precisely the metric tensor. Furthermore, a connection adapted to a given G-structure is any connection for which the characteristic tensor fields are parallel: in this case, any connection satisfying $\nabla g = 0$. In general, such connections have torsion. However, one may show that for two distinct adapted connections, their torsions are related by the Spencer differential acting on the contorsion tensor. Consequently, the cokernel of this Spencer differential characterizes a choice-independent \textit{intrinsic torsion} of a given G-structure. For the case of a Lorentzian spacetime, both the kernel and cokernel of the Spencer differential vanish, implying that there is a unique connection with vanishing torsion tensor. That is, there is a unique metric-compatible torsion-free connection: the Levi-Civita connection.

Furthermore, this argument may be utilized to determine a preferred connection on a (say, spacelike) hypersurface embedded in a Lorentzian manifold. We begin by pulling back the metric to the hypersurface. Treating this tensor as a characteristic tensor field, we thus find a new structure group on the hypersurface: $O(d-1)$. Then, applying the same procedure as in the Lorentzian case, we may construct a unique connection, which is again the Levi-Civita connection, this time for the induced metric.

Clearly, the above construction is quite natural for inducing geometric structure on a hypersurface, and so we may as well apply it to those null hypersurface embeddings $\iota :\mathcal{H}^n \hookrightarrow (M^d,g)$ that also pick out a canonical choice of $\bar{\ell}$. As we are demanding that our geometry come equipped with a preferred vector field $\bar{\ell}$ along $\mathcal{H}$, the structure group on $\iota(\mathcal{H})$ is precisely the subgroup of $O(d-1,1)$ that leaves $\iota_* \bar{\ell}$ invariant. This structure group on $\iota(\mathcal{H})$ then has two characteristic tensor fields: $\iota_* \bar{\ell}$ and $g|_{\mathcal{H}}$. Pulling these tensor fields back to $\mathcal{H}$ leaves us with two characteristic tensor fields on $\mathcal{H}$, $\bar{g}$ and $\bar{\ell}$, with structure group given by the Carroll group $O(n-1) \ltimes \mathbb{R}^{n-1}$---this will be demonstrated in Section~\ref{sec:Cstructures}. Thus, guided by the case of spacelike hypersurfaces, we look for a connection $\bar{\nabla}$ that is adapted to this structure group. Specifically, such a connection must satisfy $\bar{\nabla} \bar{g} = 0 = \bar{\nabla} \bar{\ell}$. However, as noted by~\cite{figueroa2020}, the intrinsic torsion for a Carrollian structure is nonvanishing unless $\mathcal{L}_{\bar{\ell}} \bar{g} = 0$; furthermore, the kernel of the Spencer differential is non-vanishing. We thus have the following slogan:
\begin{slogan}
A natural Koszul connection on a Carrollian structure is metric, is compatible with the fundamental vector field, and may have torsion.
\end{slogan}

In other contexts, some have found it useful to work with Koszul connections that are non-metric but torsion-free. For example, Mars~\cite{mars2013,mars2020} has extensively studied a connection for null hypersurfaces which shift $\mathcal{L}_{\bar{\ell}} \bar{g}$ into the non-metricity of the connection rather than leaving it in torsion. Others~\cite{Duval2014, Donnay_2019, Chandrasekaran2022} have found it useful to construct other Koszul connections with vanishing torsion (which we take to mean the antisymmetric piece of the connection coefficients in a coordinate basis, see Equation~(\ref{torsion})). In \cite{Ciambelli2018}, they fold the intrinsic torsion into the frame basis rather than the connection coefficients. This gives rise to additional tensorial objects on the Carrollian manifold in order to capture the information lost by using a symmetric connection. However, as argued above, the ``naturality'' of torsion-free connections is put into question~\cite{Henneaux1979}, whereas the connections we use arise in the same way that the induced Levi-Civita connection arises for a non-null embedded hypersurface.

Note that such a choice is also justified from a physical perspective. It is well-known that the Palatini formalism, applied to the Einstein-Hilbert action, reproduces the metricity condition of the Levi-Civita connection on a Lorentzian (or Riemannian) manifold. Applying the same procedure to the Carrollian limit of the Einstein-Hilbert action~\cite{musaeus2023}, one finds precisely those connections satisfying
$$\bar{\nabla} \bar{g} = 0 = \bar{\nabla} \bar{\ell}$$
with torsion vanishing torsion only when $\mathcal{L}_{\bar{\ell}} \bar{g} = 0$. Consequently, if one wished to describe the physics \textit{intrinsic} to a null hypersurface (such as fermions in a Carrollian structure, see for example~\cite{bergshoeff2023}) with non-trivial geometry, one should use a connection that arises geometrically from the manifold itself. As such, the family of Koszul connections we consider are precisely those connections that appear naturally in both the intrinsic physics of a Carrollian manifold and the extrinsic geometry of an embedded null hypersurface. In fact, we speculate that such a connection is essential for holography on a null hypersurface.

\medskip

Looming over the above discussion, however, is that we are not granted a unique Koszul connection that satisfies the conditions required. In the next section, we will go into more detail of the structure group of a Carrollian structure, as well as what is required to pick out a preferred Koszul connection from the family of adapted connections. For now, it is useful to assume such a preferred connection is given and provide one more definition.
\begin{definition}
    Let $(\mathcal{H},\bar{g},\bar{\ell})$ be a Carrollian structure equipped with a Koszul connection $\bar{\nabla}$ such that $\bar{\nabla} \bar{g} = 0$ and $\bar{\nabla} \bar{\ell} = 0$. Then, the quadruplet $(\mathcal{H},\bar{g},\bar{\ell},\bar{\nabla})$ is a \textit{Carrollian manifold} and the connection $\bar{\nabla}$ is termed a  Carrollian connection.
\end{definition}

\noindent We note that elsewhere in the literature, Carrollian structures equipped with an affine connection are sometimes called strong Carrollian geometries.

Again, as for Carrollian structures, even a non-expanding horizon could be assigned arbitrarily the geometric data necessary to describe a Carrollian manifold. However, these are non-canonical choices. As we will see in Section~\ref{sec:induced-cmanifolds}, when a natural choice of spatial submanifold exists in $\mathcal{H}$, there is a canonical choice of Carrollian connection. In \cite{ashtekar2002}, these preferred spatial submanifolds were called ``good cuts,'' and can always be found for non-extremal weakly isolated horizons.

\medskip

Via the above definitions, it is clear that a null hypersurface embedding into a Lorentzian spacetime $\iota: \mathcal{H}^n \hookrightarrow (M^d,g)$ can be realized by a unique pre-Carrollian structure, but does not uniquely determine any stronger Carrollian geometry. Consider, for example, the simplest non-trivial null hypersurface: the event horizon of a Schwarzschild black hole, given in Schwarzschild coordinates by $\mathcal{H} := \{p \in M \;|\; r(p) = 2m\}$. Motivated by~\cite{Donnay_2019}, an incautious reader might conclude that there is a single Carrollian structure associated with $\mathcal{H}$. Indeed, by setting $\bar{g} := \iota^* g$, the bilinear form is canonically determined. However, such an association implicitly assumes a natural choice of $\bar{\ell}$. In this case, there is a unique (up to constants) Killing vector field that generates the event horizon (given in Schwarzschild coordinates by $\partial_t$) which plays this role. Then the canonical identification of $\bar{\ell} := \partial_t|_{\mathcal{H}}$ yields the desired association  (see Section~\ref{sec:examples} for more details). It is important to note, however, that this is a choice---one may have chosen to identify $\bar{\ell}$ with any vector field $f(t,\theta,\phi) \partial_t|_{\mathcal{H}}$, and the resulting Carrollian structures would be distinct.

Indeed, given a vector field $\ell \in \Gamma(TM)$ that restricts to the null vector field $\bar{\ell} \in \Gamma(T \mathcal{H})$, one can consider an arbitrary null frame $(n,\ell,m_i)$ such that $n \cdot \ell = 1$ and $i \in \{3,\ldots, d\}$. Null boosts given by
\begin{align*}
n \mapsto& A^{-1} n \\
\ell \mapsto& A \ell
\end{align*}
then will, in general, correspond to \textit{distinct} Carrollian structures, related by
$$
(\mathcal{H},\bar{g},\bar{\ell}) \mapsto (\mathcal{H},\bar{g},A|_{\mathcal{H}} \bar{\ell})\,.
$$
It follows that the space of null hypersurface embeddings is much larger than the space of Carrollian structures.
Since the space of pre-Carrollian structures is so large, it is challenging to invariantly classify. (Note however that NHSs with an enlarged structure have been examined using Cartan's moving frame approach, see~\cite{nurowski2000}.) In most cases, sufficient data can be provided by the embedding geometry to canonically construct more restrictive Carrollian geometries. As such, this article will consider (briefly) the case of Carrollian structures and Carrollian manifolds in Section~\ref{sec:Cstructures}, and then we will consider in more depth the case of special Carrollian manifolds in Section~\ref{sec:induced-cmanifolds}, which naturally arise from the induced geometry of embedded null hypersurfaces such as non-extremal weakly isolated horizons \cite{ashtekar2002}.







\section{\label{sec:Cstructures} Intrinsic Carrollian Geometries \protect}

As we are interested in utilizing Carrollian geometries to characterize null hypersurfaces embedded in spacetimes, it will be useful to develop a frame formalism viewing the hypersurfaces as manifolds in their own right. The formalism on Lorentzian manifolds has led to the development of the Cartan-Karlhede algorithm which permits the local characterization of any Lorentzian manifold \cite{karlhede1980review}. This permits the classification of solutions of a given gravity theory using invariants, and can give insight into the physical properties of a solution. For example, it is conjectured that curvature invariants are able to detect the appropriate bounding hypersurface for any black hole solution as the zero-set of some invariant. This is encapsulated in a series of conjectures known as the geometric horizon conjectures~\cite{coley2017geometric}. In this section we will outline the geometric freedom or ambiguity in defining a frame formalism for Carrollian structures and manifolds.

 In a Carrollian structure $(\mathcal{H},\bar{g},\bar{\ell})$, we have very little structure to employ for classification. As mentioned in the previous section, the structure group on a Carrollian manifold is $O(n-1) \ltimes \mathbb{R}^{n-1}$. To see this, consider some basis of the tangent space, $\{\bar{\ell}, \tilde{m}_i\}$  with $\bar{g}(\tilde{m}_i, \tilde{m}_i)<0$. Then we may employ the Gram-Schmidt algorithm to find a new frame, $\{\bar{\ell}, \bar{m}_i\}$ which yields the following inner-products using the degenerate metric:
\beq g(\bar{m}_i, \bar{m}_j) = - \delta_{ij}. \eeq

\noindent This condition is not just invariant under the action of $O(n-1)$---there is an additional frame transformation that will pass between diagonal metrics. If the Gram-Schmidt algorithm is applied to the following basis $\{\bar{\ell}, \tilde{m}_i + \tilde{c}_i \bar{\ell} \}$, then the output of the algorithm will be a new frame, $\{ \bar{\ell}, \bar{m}_i + c_i \bar{\ell}\}$ where the inner product is again $\bar{g}(\bar{m}_i + c_i \bar{\ell}, \bar{m}_j + c_j \bar{\ell}) = - \delta_{ij}$. This is the additional factor of $\mathbb{R}^{n-1}$ present in the structure group.

In principle, a weak classification of Carrollian structures is possible by considering the sequence of Lie derivatives of the metric $\bar{g}$. For example, we can consider $\mathcal{L}_{\bar{\ell}} \bar{g} =: 2 K$. This tensor is in some sense horizontal as it is unaffected by shifts of the initial diagonal frame by $\bar{\ell}$-terms and its components are only affected by the group $SO(n-1)$. Using the lifted frame, $\hat{m}_i = R_i^{~j} \bar{m}_j$ where $R$ is an arbitrary element in $SO(n-1)$ we can consider directions that maximize $K(\hat{m}_i, \hat{m}_j)$ to pick out geometrically preferred spatial directions modulo $\bar{\ell}$-terms. Additional invariants could be determined by considering subsequent Lie derivatives of $K$ with respect to $\bar{\ell}$ and Lie derivatives of the maximizing spatial directions shifted by arbitrary $\bar{\ell}$-terms. It is plausible that the  degeneracy in the Gram-Schmidt procedure could also be fixed in some way. In fact, Figuera-O'Farrill already began this classification scheme for Carrollian structures, however he only considered classification up to first order~\cite{figueroa2020}. In particular, he found that Carrollian structures may be classified into four families:
\begin{enumerate}
    \item $K = 0$;
    \item $K \propto \bar{g}$;
    \item $\operatorname{tr} K = 0$;
    \item none of the above.
\end{enumerate}
Note that viewed as a $(n-1) \times (n-1)$ matrix, $\bar{g}$ is invertible, so we define $\operatorname{tr} K := \operatorname{Trace}(\bar{g}^{-1}K)$.

\medskip

Note that the additional factor of $\mathbb{R}^{n-1}$ in the structure group can also explain our inability to uniquely construct a complete coframe in the absence of some choice of the covector dual to $\bar{\ell}$.  Instead, we can build a partial coframe  using the metric, $\bar{\mu}^i = \bar{g}(\bar{m}_i,\cdot)$. This mapping from frames to the coframe is many-to-one since $\bar{\ell}$ belongs to the radical of $\bar{g}$. Weak Carrollian structures appear in Carrollian field theories in the hydrodynamic regime \cite{armas2023carrollian} where the choice of a dual to the fundamental vector field is imposed arbitrarily. Similarly, a null hypersurface embedded in a spacetime can be seen as a weak Carrollian structure if a rigging vector is not specified \cite{mars1993}.

However, given such a dual covector dual to $\bar{\ell}$, i.e. $\bar{\nu}(\bar{\ell}) = 1$, we are able to construct an associated coframe for $T^* \mathcal{H}$. It is precisely such an object that we call an \textit{Ehresmann connection}, as it defines a horizontal subbundle $H \mathcal{H}$. Furthermore, this reduces the structure group simply to $O(n-1)$ by demanding that $\bar{\nu}(\bar{m}_i) = 0$, eliminating the degeneracy in the Gram-Schmidt procedure noted above. This reduced structure group can then be used for further classification. Indeed, a wealth of invariants may be generated using Cartan's moving frame approach~\cite{olver1995}. However, we leave this potential avenue of classification for another time and do not consider it further in this article.

\medskip

We now return to the context of a Carrollian structure arising from a null hypersurface embedding in a Lorentzian spacetime. In that case, the fundamental vector field $\bar{\ell}$ is viewed as arising from a vector field $\ell \in \Gamma(T M)$. A null frame and coframe pair can then be constructed (non-uniquely) for $TM|_{\mathcal{H}}$ and $T^* M|_{\mathcal{H}}$ respectively, so that they are adapted to $\ell$. Indeed, we may use a null coframe $\{ \lambda, \nu, \mu^i\}$, where $\lambda := g(\ell,\cdot)$, so that the metric takes the form
\beq \begin{aligned}
    g = 2 \lambda \nu - \delta_{ij} \mu^i \mu^j.
\end{aligned} \eeq 
\noindent Having fixed $\lambda$, we have thus excluded both null rotations about $\nu$ and  boosts in the $(\lambda,\nu)$ plane. Due to this, the structure group  acting on the coframe adapted to a  vector field $\ell$ that is null on a hypersurface consists of:

\begin{itemize}
    \item spatial rotations: $\tilde{\mu}^i =  R^i_{~j} \mu^j$, $R^i_{~j} \in SO(n-1)$; 
    \item null rotations about $\lambda$: $\tilde{\nu} = \nu + c_i \mu^i + |c|^2 \lambda$ and $\tilde{\mu}^i = \mu^i + c^i \lambda$ for $n-1$ real-valued functions, $c_i$, where $c_i = c^i$.
\end{itemize}
\noindent These frame transformations appear in the study of degenerate Kundt spacetimes \cite{coley2009kundt} and their application to gravity theories such as quadratic gravity \cite{pravda2017exact,hervik2017universal}. Furthermore in the context of the geometric horizon conjectures, the above frame transformations are used to construct the necessary curvature invariant that detects the horizon \cite{coley2017identification}.









Carrollian structures can appear in the study of null hypersurfaces when there is no preferred spatial slice and hence a specific transverse vector field that corresponds to an Ehresmann connection. Examples of this in the literature appear for non-expanding horizons for a specific choice of $\ell$ and for an extremal weakly isolated horizon where a preferred $\ell$ direction is determined but no spatial slice arises from the associated structures on the horizon~\cite{ashtekar2002}. 

In order to develop a canonical Koszul connection for a generic null hypersurfaces, a choice must be made of either a preferred spatial slice or equivalently the transverse direction. This is exemplified in the construction of Gaussian null coordinates for an open neighbourhood of a null hypersurface \cite{moncrief1983, friedrich1999, booth2013}. In the construction of these coordinates, one must first choose coordinates for a spatial $(n-1)$-dimensional submanifold in the null hypersurface along with a vector-field in the direction of the null generator of the hypersurface. Once these choices are made there is a unique choice of a transverse null vector-field pointing off of the hypersurface whose corresponding one-form (via the musical isomorphism) acts as the Ehresmann connection on the hypersurface. The relationship between the spatial slice, the Ehresmann connection, and Gaussian null coordinates will be revisited in Theorem \ref{triad} of Section \ref{sec:induced-cmanifolds}.



\medskip

Using any Ehresmann connection, it is possible to canonically construct a Carrollian connection~\cite{bekaertmorand2018}, providing a Carrollian manifold geometry on our Carrollian structure. Consequently, it is then possible to consider the derivative of vector-fields in the manifold and compute the torsion tensor, curvature tensor and their covariant derivatives. (Note, however, that the torsion tensor depends only tensorially on the choice of Ehresmann connection.) Compared to bare Carrollian structures, Carrollian manifolds have very rigid structure, allowing for a larger set of geometric invariants.

However, there is a distinguished family of Ehresmann connections that are more useful: the \textit{principal Ehresmann connections}. These Ehresmann connections respect the Carrollian structure in that $\mathcal{L}_{\bar{\ell}} \bar{\nu} = 0$, even if these Ehresmann connections are not characteristic tensor fields of the G-structure. In the next section, we consider a specific family of such Carrollian manifolds which naturally arise from null hypersurface embeddings (which can also be shown to arise from a principal Ehresmann connection).




\section{\label{sec:induced-cmanifolds} Induced Carrollian manifolds}

As mentioned in the introduction, we must induce sufficient structure on a null hypersurface in order to describe its intrinsic and extrinsic invariants in a natural way, which can be used to characterize the aforementioned black hole solutions. Fortunately, this additional structure is available for many black hole solutions of interest, and so does not significantly restrict the applicability of this formalism. To that end, in this section, we do precisely that by examining a canonical Koszul connection and then extracting differential invariants.

As discussed in the previous section, general Carrollian structures do not pick out a unique Ehresmann connection. However, when a Carrollian structure can be induced from a null hypersurface embedding $\iota : \mathcal{H} \hookrightarrow (M,g)$ and a distinguished null vector field $\ell|_{\mathcal{H}}$ along it, the family of such Carrollian structures is greatly reduced. This is because every null hypersurface can be (locally) foliated by spatial slices, and thus there exists (at least) one integrable 1-form $\nu$ such that $\iota^* \nu \neq 0$. However, this 1-form cannot annihilate $\bar{\ell}$ anywhere, because then the kernel of $\iota^* \nu$ would have non-constant dimension. So we may demand that $\bar{\nu} = [\iota^* \nu(\bar{\ell})]^{-1} \iota^* \nu$, i.e. $\bar{\nu}(\bar{\ell}) = 1$. A simple calculation shows $\bar{\nu}$ is integrable. So, a Carrollian structure arising as an embedded null hypersurface admits an integrable Ehresmann connection.
In fact, the following results pin down these connections even further:
\begin{proposition} \label{int-closed}
    Let $(\mathcal{H},\bar{g},\bar{\ell})$ be a Carrollian structure which admits an integrable Ehresmann $\tilde{\nu}$ connection, and let $S \subset \mathcal{H}$ be a hypersurface such that for every $v \in TS$, $\tilde{\nu}(v) = 0$. Then, there exists a unique closed Ehresmann connection $\bar{\nu}$ such that $\bar{\nu}(v) = 0$.
    \end{proposition}

\begin{proof}
    Let $p \in S$, and let $V \subset S$ be a neighborhood around $p$. Then, let $\{x^i\}$ be any set of coordinates on $V$. For a sufficiently small neighborhood $U \subset \mathcal{H}$ around $p$ one can then uniquely extend these coordinates off $V$ to $U$ by keeping their values fixed along integral curves of $\bar{\ell}$. This coordinatizes $U$ via $(u,x^i)$ and yields $\bar{\ell} = \partial_u$.

    Now define $V_{t} := \{p \in U | u(p) = t\}$. By definition, $V_0 = V$, however it does not follow that for every $v \in TV_t$, we have that $\tilde{\nu}(v) = 0$. Given the canonical basis for $TV_t$ and $T^* V_t$, we can then define
    $$\bar{\nu} = \tilde{\nu} - \tilde{\nu}(\partial_{x^i}) \ext x^i\,.$$
    Since $\tilde{\nu}(\partial_{x^i})$ vanishes on $V$, we have that $\bar{\nu}|_V = \tilde{\nu}$, and so $\bar{\nu}$ is orthogonal to $TV_0$, as required. Furthermore, because $\ext x^i(\bar{\ell})=0$ by construction, we have that $\bar{\nu}(\bar{\ell}) = 1$, and so it is an Ehresmann connection. Finally, we must check that $\bar{\nu}$ is closed.

    As $\bar{\nu}$ locally generates hypersurfaces $V_t$, it is integrable and thus can be expressed as $\bar{\nu} = f \ext g$ for some functions $f,g \in C^\infty U$. Invoking the Ehresmann constraint that $\bar{\nu}(\bar{\ell}) = 1$ and that $\bar{\nu}(\partial_{x^i}) = 0$ by construction, we find that $f \ext g = \ext u$. Thus $\bar{\nu}$ is locally exact (on $U$) and hence closed on $\mathcal{H}$.

    Note that $\bar{\nu}$ is unique because any different choice of $\bar{\nu}$ would fail to preserve the Ehresmann condition or the condition that $\bar{\nu}$ preserves $S$.
\end{proof}

\noindent 
During preparation of this article, it was brought to our attention that this result is a special case of~\cite{manzano2023}.

As a result of the above, there in fact exists a canonical principal Ehresmann connection:
\begin{corollary} \label{int-princ}
    Let $(\mathcal{H},\bar{g},\bar{\ell})$ be a Carrollian structure which admits an integrable Ehresmann connection $\tilde{\nu}$ which is orthogonal to a hypersurface $S \subset \mathcal{H}$. Then the Carrollian structure admits a canonical principal Ehresmann connection $\bar{\nu}$ that is orthogonal to $S$.
\end{corollary}
\begin{proof}
    From the proposition above, if $\tilde{\nu}$ is an integrable Ehresmann connection on a Carrollian structure $(\mathcal{H},\bar{g},\bar{\ell})$, then it admits a closed Ehresmann connection $\bar{\nu}$ that is orthogonal to $S$. From Cartan's magic formula,
    $$\mathcal{L}_{\bar{\ell}} \bar{\nu} = i_{\bar{\ell}} \ext \bar{\nu} + \ext i_{\bar{\ell}} \bar{\nu} = 0\,,$$
    because $\ext \bar{\nu} = 0$ and $i_{\bar{\ell}} \bar{\nu} = 1$. Thus $\bar{\nu}$ is principal.    
\end{proof}

As a consequence of Proposition~\ref{int-closed} and Corollary~\ref{int-princ}, it is clear that those Carrollian structures arising from null hypersurface embeddings should be (and can always be) prescribed a \textit{principal} Ehresmann connection.

Recall from the introduction that the Carrollian geometries do not, in general, have a uniquely defined torsion-free connection \cite{figueroa2020,Campoleoni2022}. However, given a principal Ehresmann connection, Bekaert and Morand \cite{bekaertmorand2018} showed that one may construct a canonical Carrollian connection.

The construction is as follows. First observe that, up to spatial rotations, we have a canonical frame and coframe, given by $(\bar{\ell}, \bar{m}_i)$ and $(\bar{\nu}, \bar{\mu}^i)$, respectively, where $\bar{\mu}^i(\bar{\ell}) = 0$ and $\bar{\nu}(\bar{m}_i) = 0$. Furthermore, we have a projector to the horizontal vector (and form) bundles: namely, $q_a^b = \delta_a^b - \bar{\ell}^b \bar{\nu}_a$. By construction, we have that $q^b_{~a} \bar{\ell}^a = 0 = q^b_{~a} \bar{\nu}_b$. Furthermore, a partial inverse $\bar{g}^{ab}$ to the degenerate metric $\bar{g}$ can be constructed via the relations $\bar{g}^{ab} \bar{g}_{bc} = q^a_{~b}$. This partial inverse can be used to raise indices of horizontal 1-forms---that is, if $\omega_a \bar{\ell}^a = 0$, then we can write $\omega^a = \bar{g}^{ab} \omega_b$.
Using these objects, we can define the canonical Carrollian connection  $\bar{\nabla}$ by the coordinate expression
\begin{align}
\label{carroll-conn} 
\begin{split}\bar{\Gamma}^{\lambda}_{~\rho \sigma} = \bar{\ell}^\lambda \partial_{(\rho} \bar{\nu}_{\sigma)} + \tfrac{1}{2} \bar{g}^{\lambda \mu} [\partial_{\rho} \bar{g}_{\sigma \mu} + \partial_{\sigma} \bar{g}_{\rho \mu} - \partial_{\mu} \bar{g}_{\rho \sigma}] \\- \tfrac{1}{2} \bar{g}^{\lambda \mu} \bar{\nu}_{\sigma} \mathcal{L}_{\bar{\ell}} \bar{g}_{\rho \mu}\,.
\end{split}
\end{align}
One nice feature of this Carrollian connection is that, when $\bar{\nu}$ is closed, $\bar{\nu}$ is parallel, i.e. $\bar{\nabla} \bar{\nu} = 0$.

It is thus clear that when a preferred spatial slice of a null hypersurface is distinguished from the Lorentzian geometry, an induced Carrollian structure $(\mathcal{H},\bar{g},\bar{\ell})$ picks out a distinguished closed (and thus principal) Ehresmann connection $\bar{\nu}$ and hence a distinguished Carrollian manifold $(\mathcal{H},\bar{g},\bar{\ell},\bar{\nabla})$ with a connection that renders $\bar{\nu}$ parallel. We call such an induced Carrollian manifold specified by the tuple $(\mathcal{H},\bar{g},\bar{\ell},\bar{\nu},\bar{\nabla})$ a \textit{special Carrollian manifold}.

\begin{remark}
Observe that the family of Carrollian manifolds induced from the data described above is much smaller than the family of all Carrollian manifolds with the same Carrollian structure. Indeed, the Carrollian manifolds that arise in this way are specifically those for which there exists a covariantly constant dual to the fundamental vector field.
\end{remark}

\medskip

Having established the geometric structures required to canonically determine a special Carrollian manifold from a hypersurface, we now provide an equivalence of three different ways this data may manifest itself in the bulk geometry.

\begin{theorem} \label{triad}
Let $\iota: \mathcal{H}^{d-1} \hookrightarrow (M^d,g)$ be a null hypersurface embedding into a Lorentzian manifold. Let $p \in \mathcal{H}$ and let $U \subset M$ be a small neighborhood around $p$. Then the following are canonically equivalent:
\begin{itemize}
\item A defining function $r$ for $\mathcal{H} \cap U$ paired with a null vector $n \in \Gamma(T M)|_{\mathcal{H} \cap U}$ transverse to $\mathcal{H}$ such that the one-form $\bar{\nu} := \iota^* g(n,\cdot)$ is closed and $n(r)|_{\mathcal{H} \cap U} = 1$.
\item A null vector field $\bar{\ell} \in \Gamma(T (\mathcal{H} \cap U))$ paired with a spacelike codimension-1 submanifold $\mathcal{S} \subset \mathcal{H} \cap U$ containing $p$;
\item A Gaussian null coordinatization of $U$.
\end{itemize}
\end{theorem}

\begin{proof}
We prove by cycling.
We begin by defining $\ell := g^{-1}(\ext r, \cdot)|_{U}$. Since  $\mathcal{H}$ is null, $\ext r|_{\mathcal{H} \cap U}$ is null, and hence so is $\ell|_{\mathcal{H} \cap U}$. It then follows that $\ext r(\ell)|_{\mathcal{H} \cap U} = 0$, and so there exists $\bar{\ell} \in \Gamma(T (\mathcal{H} \cap U))$ such that $\ell|_{\mathcal{H} \cap U} = \iota_* \bar{\ell}$. 
Now as a non-vanishing 1-form on $\mathcal{H} \cap U$, $\bar{\nu}$ has a $d-2$ dimensional kernel. Furthermore, because $\bar{\nu}(\bar{\ell}) = n(r)|_{\mathcal{H} \cap U} =1$, we have that $\bar{\ell} \not \in \ker \bar{\nu}$, and thus $\ker \bar{\nu}$ is spacelike. Now observe that because $\bar{\nu}$ is an integrable Ehresmann connection, from Proposition~\ref{int-closed}, it is closed. Now let $x,y \in \ker \bar{\nu}$. Then, because $\bar{\nu}$ is closed,
$$0 = \ext \bar{\nu}(x,y) = x(\bar{\nu}(y)) - y(\bar{\nu}(x)) - \bar{\nu}([x,y])\,,$$
so it follows that $[x,y] \in \ker \bar{\nu}$. Thus we have a spacelike integrable $d-2$ distribution, which by the Frobenius theorem yields the desired spacelike submanifold $\mathcal{S} \subset \mathcal{H} \cap U$. So we have reached the second bullet point.

The third bullet point follows from a direct higher-dimensional generalization of Friedrich, et. al.~\cite{friedrich1999}. Thus we obtain a set of coordinates on $U$ given by $(v,u,x^A)$ such that
$$\ext s^2 = 2 \ext u \ext v + A \ext u^2 + B_A \ext u \ext x^A + C_{AB} \ext x^A \ext x^B\,,$$
where $A|_{\mathcal{H}} = 0 = B_A|_{\mathcal{H}}$.

Finally, given a Gaussian null coordinates $(v,u,x^A)$ on $U$ as above, we make the canonical choices $r := v$ and $n := \partial_{v}|_{\mathcal{H}}$. Clearly, $n$ is transverse to $\mathcal{H}$ and null. Furthermore, we find explicitly that $g(n,\cdot) = \ext u$, a closed one form. Thus it follows that $\bar{\nu} := \iota^*g(n,\cdot)$ is closed. We have thus obtained the first bullet point, which completes the proof.
\end{proof}

This theorem shows that these three ways of specifying the additionally required data are equivalent, and they all provide the same closed (and hence principal) Ehresmann connection $\bar{\nu}$ as well as a fundamental vector $\bar{\ell}$ on the null hypersurface with induced (degenerate) metric $(\mathcal{H},\bar{g})$. This information is sufficient to fully specify a special Carrollian manifold $(\mathcal{H},\bar{g},\bar{\ell}, \bar{\nu},\bar{\nabla})$, and it is such induced structures that we will consider going forward. 

\medskip

Before we proceed, it is worthwhile to take stock of what has been accomplished here. Given a null hypersurface $\mathcal{H}$ embedded in $(M,g)$, we have canonically induced sufficient structure on $\mathcal{H}$ to regard it as a Carrollian manifold in a geometrically natural way. This is entirely analogous, if more complex, than the procedure for spacelike (or timelike) hypersurfaces. There, one can induce a Riemannian (or pseudo-Riemannian) metric, which in turn induces a Levi-Civita connection, giving the hypersurface a (pseudo-)Riemannian structure in its own right. Having done this, one can then begin to construct intrinsic and extrinsic differential invariants of the embedding, such as the second fundamental form.

Having suceeded at a similar task in the case of an embedded null hypersurface, we are now situated to study similarly intrinsic and extrinsic differential invariants of the embedding. In principle, these invariants can be used to further study null hypersurface embeddings via classification methods, such as one similar to the Cartan--Karlhede method.

\subsection{Intrinsic Differential Invariants}
We begin by considering the intrinsic differential invariants associated to a special Carrollian manifold $(\mathcal{H},\bar{g},\bar{\ell},\bar{\nu},\bar{\nabla})$. Given the connection in Equation~(\ref{carroll-conn}), the torsion tensor is given by
\begin{align}\label{torsion}
\bar{T}^a_{~bc} = 
-\bar{g}^{ad} \bar{\nu}_{[b} \mathcal{L}_{\bar{\ell}} \bar{g}_{c]d}\,.
\end{align}
Of particular use is the so-called \textit{intrinsic second fundamental form}, given by
$$K_{ab} = \tfrac{1}{2} \mathcal{L}_{\bar{\ell}} \bar{g}_{ab}\,.$$
As noted in Section~\ref{sec:Cstructures}, $K$ is horizontal.
Thus, in terms of this tensor, we have that $\bar{T}^a_{bc} = 2 \bar{\nu}_{[b} K^a_{c]}$.

Next, we can express the curvature tensor in abstract index notation via
$$\bar{R}_{ab}{}^c{}_{d} Z^{d} = ([\bar{\nabla}_a, \bar{\nabla}_b] + T^{e}_{ab} \bar{\nabla}_e) Z^c\,.$$
Rewriting this in terms of the intrinsic second fundamental form, we have that
$$\bar{R}_{ab}{}^c{}_{d} Z^{d} = ([\bar{\nabla}_a, \bar{\nabla}_b] + 2 \bar{\nu}_{[a} K^e_{b]} \bar{\nabla}_e) Z^c\,.$$

Now, in more general contexts with less constrained structure, one can extract invariants by applying invariant operations to those structures. For example, one might consider producing additional differential invariants by evaluating the exterior derivative of the Ehresmann connection (these are called the Carrollian torsion and acceleration in~\cite{Ciambelli2019}). However, by construction, for a special Carrollian manifold, the Ehresmann connection is closed, and thus no additional curvature invariants appear. As noted in~\cite{Ciambelli2019}, this provides additional confirmation that such a principal Ehresmann connection produces an integrable horizontal bundle. Note that this is a stronger condition than that imposed by~\cite{Ciambelli2019}, but it is nonetheless always the case for special Carrollian manifolds.

Since Lie derivatives with respect to $\bar{\ell}$ can be expressed in terms of $\bar{\nabla}$ and $\bar{T}$ (which in turn can be expressed in terms of $\bar{\nu}$ and $K$), it is now clear that any intrinsic differential invariant of the  special  Carrollian manifold can be expressed via a tensor built from the following set of ingredients:
\begin{align} \label{intrinsic-list}
\{ \bar{g},\bar{\ell},\bar{\nu},\bar{\nabla},\bar{R},K\}\,.
\end{align}

\subsection{Bundle Decomposition on $\mathcal{H}$}

In order to relate spacetime invariants to their hypersurface counterparts, it is essential that we provide a distinguished decomposition of the tangent and cotangent bundles at a point $p \in \mathcal{H} \subset M$. We will assume that a set of data specified by Theorem~\ref{triad} is given and that it determines the special Carrollian manifold $\{\mathcal{H},\bar{g},\bar{\ell},\bar{\nu},\bar{\nabla}\}$. Now, consider a frame of $T_p \mathcal{H}$ which is canonically decomposed into vertical and horizontal components: $\{\bar{\ell}, \bar{m}^i\}$. Here, $\{\bar{m}^i\}$ form any orthonormal frame for the horizontal subspace given by $H \mathcal{H} := \ker \bar{\nu} \subset T \mathcal{H}$. As $\bar{\nu}$ is an Ehresmann connection, it is dual to $\bar{\ell}$ and hence, together with $\bar{g}$, we have a coframe $\{\bar{\nu},\bar{\mu}_i\}$.

We would like to align a distinguished frame in the spacetime at the point $p$ with that given for $T_p \mathcal{H}$. At any such point, there is a unique null rigging vector~\cite{mars1993} $n \in T_p M$ that satisfies $g(n,\iota_* \bar{\ell}) = 1$ and $g(n,\iota_* \bar{m}^i) = 0$ that is transverse to $\mathcal{H}$. (Note that $\iota^* g(n,\cdot) = \bar{\nu}$.) This vector can be used to provide a distinguished decomposition
\begin{align} \label{partial-decomp} T_p M = \langle n_p \rangle \oplus T_p \mathcal{H}\,.
\end{align}
In fact, this decomposition induces a surjection $T : TM|_{\mathcal{H}} \rightarrow T \mathcal{H}$ so that for any $v \in TM|_{\mathcal{H}}$, we have that $T(v) \in T \mathcal{H}$. Note that $T$ is the left-inverse of the pushfoward by $\iota$, in the sense that $(T \circ \iota_*)(\bar{v}) = \bar{v}$, for any $\bar{v} \in T \mathcal{H}$. Further, for $v \perp n$ in the sense that under the decomposition given, $v$ has a vanishing coefficient for $n$, we have that $(\iota_* \circ T)(v) = v$, and thus $T$ is a partial right-inverse of $\iota_*$.

The decomposition given in Equation~(\ref{partial-decomp}) can be further refined by the pushfoward of the frame $\iota_* \{\bar{\ell},\bar{m}^i\} := \{\ell,m^i\}$. We thus achieve
$$T_p M = \langle n_p \rangle \oplus \langle \ell_p \rangle \oplus H_p \mathcal{H}\,.$$
This decomposition also agrees with the decomposition of $T \mathcal{H}$ because $T(n) = 0$, $T(\ell) = \bar{\ell}$, and $T(m^i) = \bar{m}^i$.

A similar distinguished decomposition can be constructed for the cotangent bundle at a point $p \in \mathcal{H}$. Defining $\lambda := g(\ell,\cdot)$, $\nu := g(n,\cdot)$, and $\mu_i := g(m^i,\cdot)$, we find a distinguished decomposition
$$T^*_p M = \langle \lambda_p \rangle \oplus \langle \nu_p \rangle \oplus H^*_p\mathcal{H}\,.$$
This decomposition, too, agrees with the decomposition on the special Carrollian manifold, because $\iota_*(\lambda) = 0$, $\iota_*(\nu) = \bar{\nu}$, and $\iota_*(\mu_i) = \bar{\mu}_i$.

With these decompositions in hand, we can now begin to both relate connections and construct invariants of the embedding.

\subsection{Hypersurface Calculus}
In principle, every intrinsic invariant of a special Carrollian manifold $\{\mathcal{H},\bar{g},\bar{\ell},\bar{\nu},\bar{\nabla}\}$ can be built from Set~(\ref{intrinsic-list}). However, as we are interested in studying those that arise from null hypersurface embeddings and data specified as in Theorem~\ref{triad}, it is imperative that we determine the relationship between these objects and the defining objects of the bulk given by $\{g, \nabla, R, \nu,\ell \}$. Clearly, $\iota^* g = \bar{g}$, and the relationships for $\nu$ and $\ell$ were given in the previous section.

We can now consider the first non-trivial calculation, which involves finding a spacetime expression for $K$. As $T$ is a partial right-inverse of $\iota_*$, we can express $T(\ell)$ as $\iota^* \ell$, where $\iota^*$ here is the inverse of the pushforward map. Thus, we can write
\begin{align*}
 K =& \tfrac{1}{2}\mathcal{L}_{\iota^* \ell} \iota^* g \\
=&  \tfrac{1}{2}(\ext i_{\iota^* \ell} + i_{\iota^* \ell} \ext ) \iota^* g \\
=&  \tfrac{1}{2}\iota^* (\ext i_{\ell} + i_{\ell} \ext) g \\
=&  \tfrac{1}{2}\iota^* \mathcal{L}_{\ell} g \\
=& \iota^* \odot\nabla \lambda\,.
\end{align*}
Having established a spacetime formula for the intrinsic second fundamental form, we next turn to establishing the relationship between the spacetime Levi-Civita connection and the induced Carrollian connection.

\subsection{Spacetime-Carrollian relationships: the connection}

Recall the distinguished decomposition of $T_{p \in \mathcal{H}} M$:
$$T_p M = \langle n_p \rangle \oplus \langle \ell_p \rangle \oplus H_p \mathcal{H}\,.$$
Restricting the Levi-Civita connection to $\mathcal{H}$, we can write
$$\nabla^\top : \iota_* \Gamma(T \mathcal{H}) \otimes \Gamma(TM)|_{\mathcal{H}} \rightarrow \Gamma(TM)|_{\mathcal{H}}\,.$$
In order to describe the spacetime invariants in terms of intrinsic invariants of $\mathcal{H}$ and extrinsic invariants of the embedding $\mathcal{H} \hookrightarrow M$, we must decompose this connection and examine the projection onto each component. Specifically, for $\bar{u},\bar{v},\bar{w} \in \Gamma(H\mathcal{H})$, we consider the following scalars:
\begin{equation} \label{raw-table} \begin{aligned} 
\lambda_a &\nabla^\top_{\ell} n^a \,, \quad &\nu_a &\nabla^\top_{\ell} n^a \,, \quad &&(\iota_* \bar w)_a \nabla^\top_{\ell} n^a\,,\\
\lambda_a &\nabla^\top_{\ell} \ell^a\,, \quad &\nu_a &\nabla^\top_{\ell} \ell^a \,, \quad &&(\iota_* \bar w)_a \nabla^\top_{\ell} \ell^a \,, \\
\lambda_a &\nabla^\top_{\ell} (\iota_* \bar v)^a \,, \quad &\nu_a &\nabla^\top_{\ell} (\iota_* \bar v)^a\,, \quad &&(\iota_* \bar w)_a \nabla^\top_{\ell} (\iota_* \bar v)^a\,, \\
\lambda_a &\nabla^\top_{\iota_* \bar u} n^a \,, \quad &\nu_a &\nabla^\top_{\iota_* \bar u} n^a \,, \quad &&(\iota_* \bar w)_a \nabla^\top_{\iota_* \bar u} n^a\,,\\
\lambda_a &\nabla^\top_{\iota_* \bar u} \ell^a\,, \quad &\nu_a &\nabla^\top_{\iota_* \bar u} \ell^a \,, \quad &&(\iota_* \bar w)_a \nabla^\top_{\iota_* \bar u} \ell^a \,, \\
\lambda_a &\nabla^\top_{\iota_* \bar u} (\iota_* \bar v)^a \,, &\quad \nu_a &\nabla^\top_{\iota_* \bar u} (\iota_* \bar v)^a\,, \quad &&(\iota_* \bar w)_a \nabla^\top_{\iota_* \bar u} (\iota_* \bar v)^a\,.
\end{aligned}
\end{equation}

\subsubsection{Components of $\nabla^\top$ in the $n$ direction}
We begin by computing, in terms of intrinsic and extrinsic objects, the derivatives in the left column of table in Equation~\ref{raw-table}. Observe that, because $\lambda(n) \eqSig 1$ everywhere on $\mathcal{H}$, we can write $\lambda_a \nabla^\top_{\ell} n^a = -\nu_{a} \nabla^\top_{\ell} \ell^a$. Restricted to the hypersurface $\mathcal{H}$, this is the projection of $\nabla^\top_{\ell} \ell^a$ to the $\ell$ direction. However, because $\ell|_{\mathcal{H}}$ generates a null hypersurface, we necessarily have that  $\nabla^\top_{\ell} \ell |_{\Sigma} \propto \ell$, with proportionality constant given by the surface gravity of the hypersurface, denoted by $\kappa$. And so, we have that
$$\lambda_a \nabla^\top_{\ell} n^a = -\kappa\,.$$

Next we consider $\lambda_a \nabla^\top_{\iota_* \bar u} n^a$. This derivative measures how the vector $n$ changes in affine length as it is moved along the horizontal submanifold. In the physics literature, this is called the H\'aj\'i\v{c}ek one-form~\cite{hajicek1973,hajicek1975,gourgoulhon2005} (or in the mathematics literature, the normal fundamental form~\cite[Volume 4]{spivak}). We write
$$\beta_a := q^c_{~a} \iota^*(\lambda_b \nabla n^b)_c \in \Gamma(H\mathcal{H})\,,$$
so that
$$\lambda_a \nabla^\top_{\iota_* \bar u} n^a = \bar{u}^a \beta_a\,.$$

We now consider terms involving $\lambda$ with no dependence on $n$---that is, we consider terms of the form $\lambda_a \nabla^\top_{\iota_* \bar{x}} (\iota_* \bar{y})^a$, where $\bar{x},\bar{y} \in T \mathcal{H}$, and $x,y$ are their pushforwards, respectively. As the Levi-Civita connection is linear both in $x$ and $y$, and we have that
$$\lambda_a \nabla^\top_x f y^a = \lambda_a \nabla^\top_{fx} y^a =  f \lambda_a \nabla^\top_x y^a$$
for all $f \in C^\infty M$, following Spivak~\cite[Volume 1, Section 4, Theorem 2]{spivak} and~\cite[Volume 3, Section 1, Theorem 5]{spivak}, we have that there exists a symmetric tensor field $s : \Gamma(T \mathcal{H}) \times \Gamma(T \mathcal{H}) \rightarrow \Gamma(\Span{n})$. Hence we can write
$$\lambda_c s_{~ab}^c \bar{x}^a \bar{y}^b  = -\lambda_a \nabla^\top_{\iota_* \bar{x}} (\iota_* \bar{y}^a)\,.$$
As $\lambda (\iota_* \bar{y}) = 0$, we can re-express this tensor as
$$\lambda_c s^c_{~ab} = \iota^* (\nabla_{(a} \lambda_{b)})\,.$$
It follows that
$$\lambda_c s^c_{~ab} = K_{ab}\,.$$
We have thus established that
$$\lambda_a \nabla^\top_{\iota_* \bar{x}} (\iota_* \bar{y})^a = -K_{ab} \bar{x}^a \bar{y}^a\,.$$
However, $K$ is horizontal, so the following four values can be filled into the left column of table in Equation~\ref{raw-table}:
\begin{align*}
\lambda_a \nabla^\top_{\ell} \ell^a &= 0 \\
\lambda_a \nabla^\top_{\ell} (\iota_* \bar{v})^a &= 0 \\
\lambda_a \nabla^\top_{\iota_* \bar{u}} \ell^a &= 0 \\
\lambda_a \nabla^\top_{\iota_* \bar{u}} (\iota_* \bar{v})^a &= -K_{ab} \bar{u}^a \bar{v}^b\,.
\end{align*}

\subsubsection{Components of $\nabla^\top$ in the $\ell$ direction}
We next compute the derivatives in the middle column of table in Equation~\ref{raw-table} in much the same way as in the previous subsection. As $g(n,n) \eqSig 0$, from the Leibniz property and metric compatibility, we have that
$$\nu_a \nabla_b^\top n^a = 0\,.$$

By definition of surface gravity, we have that
$$\nu_a \nabla^\top_{\ell} \ell^a = \kappa\,.$$

Next, consider  $\nu_a \nabla^\top_{\iota_* \bar u} \ell^a$. By applying the Leibniz rule to the definition of the  H\'aj\'i\v{c}ek one-form and the fact that $\lambda(n) \eqSig \nu(\ell) \eqSig 1$, we have that
\begin{align} \label{neg-haji}
\nu_a \nabla^\top_{\iota_* \bar u} \ell^a = -\bar{u}^a \beta_a\,.
\end{align}

Using this, we consider $\nu_a \nabla^\top_{\ell} (\iota_* \bar v)^a$. We evaluate below:
\begin{align*}
\nu_a \nabla^\top_{\ell} (\iota_* \bar v)^a &\eqSig \nu_a ([\iota_* \bar \ell, \iota_* \bar v]^a + \nabla_{\iota_* \bar v} \ell^a) \\
&\eqSig (\iota^* \nu)([\bar \ell, \bar v]) + \nu_a \nabla^\top_{\iota_* \bar v} \ell^a \\
&\eqSig \bar{\nu}([\bar \ell, \bar v]) - \bar{v}^a \beta_a \\
&\eqSig \bar{\nu}_a (\bar{\nabla}_{\bar \ell} \bar{v}^a - \bar{\nabla}_{\bar v} \bar{\ell}^a - \bar{T}^a_{bc} \bar{\ell}^b \bar{v}^c) - \bar{v}^a \beta_a \\
&\eqSig \bar{\nu}_a \bar{\nabla}_{\bar \ell} \bar{v}^a  - \bar{v}^a \beta_a \\
&\eqSig  - \bar{v}^a \beta_a\,.
\end{align*}
In the second equality, we used the standard result that $\iota_* [\bar{u},\bar{v}] = [\iota_* \bar{u}, \iota_* \bar{v}]$, and in the third equality, we used Equation~(\ref{neg-haji}). In the fourth equality, we used the standard expression for the Lie bracket in terms of a torsionful connection, and in the fifth equality we used that $\bar{\nabla} \bar{\ell} = 0$ and that $\bar{\nu}_a \bar{T}^a_{~bc} = 0$. In the sixth line we used the Leibniz rule, that $\bar{\nu}(\bar v) = 0$, and that $\bar{\nabla} \bar{\nu} = 0$.

We finish this subsection by examining $\nu_a \nabla^\top_{\iota_* \bar u} (\iota_* \bar v)^a$. Similar to the previous subsection, we observe that for horizontal vector fields $\bar{u},\bar{v} \in H\mathcal{H}$, we have that $\nu_a \nabla^\top_{\iota_* \bar u} f (\iota_* \bar v)^a = f \nu_a \nabla^\top_{\iota_* \bar u} (\iota_* \bar v)^a$ for $f \in C^\infty M$. So again following Spivak, there exists a symmetric tensor field $\II : \Gamma(H\mathcal{H}) \times \Gamma(H\mathcal{H}) \rightarrow \Gamma(\Span{\ell})$, labelled the \textit{extrinsic second fundamental form} (or sometimes just the second fundamental form, for short). So, we can write
\begin{align*}
 \II_{ab} \bar{u}^a \bar{v}^b :=& -\nu_a \nabla^\top_{\iota_* \bar{u}} (\iota_* \bar{v}^a) \\
\eqSig& (\iota_* \bar{v})^a (\iota_* \bar{u})^b \nabla_{(a} \nu_{b)} \\
\eqSig& \tfrac{1}{2} (\iota_* \bar{v})^a (\iota_* \bar{u})^b \mathcal{L}_n g_{ab} \\
\eqSig& \tfrac{1}{2} (\iota^* \mathcal{L}_n g_{ab})(\bar{u},\bar{v})\,.
\end{align*}
Thus, we have that
$$\II_{ab} = \tfrac{1}{2} \iota^* \mathcal{L}_n g_{ab}\,.$$
With this tensor in hand, we thus have that
$$\nu_a \nabla^\top_{\iota_* \bar u} (\iota_* \bar v)^a = -\II_{ab} \bar{u}^a \bar{v}^b\,.$$

\subsubsection{Components of $\nabla^\top$ in the horizontal direction}
We finish this series of computations by considering the last column in table in Equation~\ref{raw-table}. The simplest of these derivatives follows from the geodesicity of $\ell$, namely $(\iota_* \bar w)_a \nabla^\top_{\ell} \ell^a = 0$. Beyond this, two others are already easily computed using the Leibniz rule and methods already established:
$$(\iota_* \bar w)_a \nabla^\top_{\iota_* \bar u} n^a = \II_{ab} \bar{u}^a \bar{w}^b$$
and
$$(\iota_* \bar w)_a \nabla^\top_{\iota_* \bar u} \ell^a = K_{ab} \bar{u}^a \bar{w}^b\,.$$

Next consider $(\iota_* \bar{w})_a \nabla^\top_{\ell} n^a$. Since $\nabla^\top_{\ell} \nu(\iota_* \bar w) = 0$, we have that
$$(\iota_* \bar{w})_a \nabla^\top_{\ell} n^a = \bar{w}^a \beta_a\,. $$

Now consider terms of the form $(\iota_* \bar w)_a \nabla^\top_{\iota_* \bar u} (\iota_* \bar v)^a$. To do so, we generalize a similar result found in~\cite[Volume 3]{spivak}, and hence consider the sum:
\begin{widetext}
\begin{align*}
\bar{g}_{jk} \bar{\nabla}_{\bar u} \bar{v}^j \bar{w}^k + \bar{g}_{jk} \bar{\nabla}_{\bar v} \bar{u}^j \bar{w}^k -  \bar{g}_{jk} \bar{\nabla}_{\bar w} \bar{u}^j \bar{v}^k 
=& \bar{g}_{jk} \bar{w}^k (\bar{\nabla}_{\bar u} \bar{v}^j + \bar{\nabla}_{\bar v} \bar{u}^j)  +  \bar{g}_{jk} \bar{v}^k (\bar{\nabla}_{\bar u} \bar{w}^j - \bar{\nabla}_{\bar w} \bar{u}^j) +  \bar{g}_{jk} \bar{u}^k (\bar{\nabla}_{\bar v} \bar{w}^j - \bar{\nabla}_{\bar w} \bar{v}^j) \\
=& \bar{g}_{jk} \bar{w}^k (2 \bar{\nabla}_{\bar v} \bar{u}^j + \bar{T}^j_{~im} \bar{u}^i \bar{v}^m + [\bar{u},\bar{v}]^j)  +  \bar{g}_{jk} \bar{v}^k (\bar{T}^j_{~im} \bar{u}^i \bar{w}^m + [\bar{u},\bar{w}]^j) \\&+ \bar{g}_{jk} \bar{u}^k (\bar{T}^j_{~im} \bar{v}^i \bar{w}^m + [\bar{v},\bar{w}]^j)\,.
\end{align*}
\end{widetext}
Now let $u_p = \iota_* \bar{u}_p$, $v_p = \iota_* \bar{v}_p$, and $w_p = \iota_* \bar{w}_p$, so that $u,v,w \in TM|_{\mathcal{H}}$. As before, we can use $\iota^*$ as the inverse of the pushforward map here because these are horizontal vectors. (Note that for what follows, we do not need to know how $u,v,w$ extend into $M$ because we are only taking derivatives in directions tangent to $\mathcal{H}$.) Performing the same computation with the spacetime Levi-Civita connection, a similar computation follows except the torsion vanishes. However, we also have that
\begin{align*}
\bar{g}_{jk} \bar{\nabla}_{\bar u} \bar{v}^j \bar{w}^k =& \bar{u} \langle \bar{v},\bar{w} \rangle_{\bar{g}} \\
=& i_{\bar{u}} \ext \langle \bar{v}, \bar{w} \rangle_{\bar{g}} \\
=& i_{\iota^* u} \ext \langle \iota^* v, \iota^* w \rangle_{\iota^* g} \\
=& i_{\iota^* u} \ext \iota^* \langle v,w \rangle_g \\
=& \iota^* i_u \ext \langle v,w \rangle_g \\
=& g_{ab} \nabla^\top_u v^a w^b\,.
\end{align*}
We need to compute a similar result for the Lie brackets. However, note that because $\mathcal{H}$ is a hypersurface, we have that $[\bar{v},\bar{w}] =: \bar{x} \in T \mathcal{H}$. So in particular there exists $x \in TM|_{\mathcal{H}}$ such that $\iota^* x = \bar{x}$. Viewing $\iota$ as a diffeomorphism $\mathcal{H} \rightarrow M|_{\mathcal{H}}$, we have that
$$x = \iota_* \bar{x} = \iota_* [\bar{v},\bar{w}] = [\iota_* \bar{v}, \iota_* \bar{w}] = [v,w]\,.$$
We now can compute at $p \in \mathcal{H} \subset M$:
\begin{align*}
\langle \bar{u}, [\bar{v},\bar{w}] \rangle_{\bar g_p} =& \langle \iota^* u, \iota^* [v,w] \rangle_{\iota^* g_p} \\
=& (\iota^* g)_p(\iota^* u, \iota^*[v,w]) \\
=& g_p(u,[v,w]) \\
=& \langle u, [v,w] \rangle_{g_p} \,.
\end{align*}
Combining these three computations, we have that
\begin{align*}
g_{ab} w^b \nabla^\top_v u^a = \bar{g}_{ab} (& \bar{w}^b \bar{\nabla}_{\bar v} \bar{u}^a + \tfrac{1}{2} \bar{w}^b \bar{T}^a_{~cd} \bar{u}^c \bar{v}^d \\&+ \tfrac{1}{2} \bar{v}^b \bar{T}^a_{~cd} \bar{u}^c \bar{w}^d + \tfrac{1}{2} \bar{u}^b \bar{T}^a_{~cd} \bar{v}^c \bar{w}^d)\,.
\end{align*}

Expressing this relationship in terms of the intrinsic second fundamental form and noting that $\bar{\nu}(\bar{u}) = \bar{\nu}(\bar{v}) = \bar{\nu}(\bar{w}) = 0$, we have that
$$(\iota_* \bar w)_a  \nabla^\top_{\iota_* \bar u} (\iota_* \bar v)^a = \bar{w}_a \bar{\nabla}_{\bar u} \bar{v}^a\,.$$

We conclude this subsection by considering terms of the form $(\iota_* \bar w)_a \nabla^\top_{\ell} (\iota_* \bar v)^a$. To do so, we will consider the intrinsic derivative $\bar{w}_a \bar{\nabla}_{\bar \ell} \bar{v}^a$. As $\bar{\nabla} \bar{\ell} = 0$ and $\bar{g}(\bar{\ell},\cdot) = 0$, we can express this derivative in terms of a Lie bracket and the torsion tensor:
\begin{align*}
\bar{w}_a \bar{\nabla}_{\bar{\ell}} \bar{v}^a =& \bar{w}_a [\bar{\ell},\bar{v}]^a +  \bar{w}_a \bar{\ell}^c \bar{v}^d \bar{T}^a_{~cd} \\
=& \bar{w}_a [\bar{\ell},\bar{v}]^a + 2 \bar{\nu}_{[c} K_{~d]}^a \bar{\ell}^c \bar{v}^d \\
=& \bar{w}_a [\bar{\ell},\bar{v}]^a + \bar{w}^a K_{ab} \bar{v}^b\,. 
\end{align*}
Now for the same reasons as before, we have that $\iota_* [\bar{\ell},\bar{v}] = [\ell,v]$, where $v = \iota_* \bar{v}$. As before, let $w = \iota_* \bar{w}$. Thus, we have that
$$\bar{w}_a [\bar{\ell},\bar{v}]^a = w_a [\ell,v]^a\,.$$
Now calculating in $M$ on $\mathcal{H}$, we have that
\begin{align*}
w_a [\ell,v]^a \eqSig& w_a \nabla^\top_{\ell} v^a - w^a v^b \nabla_b \lambda_a \\
\eqSig& w_a \nabla^\top_{\ell} v^a - w^a v^b \nabla_{(a} \lambda_{b)} - w^a v^b (\ext  \lambda)_{ab} \\
\eqSig& w_a \nabla^\top_{\ell} v^a - (\odot \nabla \lambda)(\iota_* \bar{v}, \iota_* \bar{w}) - (\ext \lambda)(\iota_* \bar{w},\iota_* \bar{v}) \\
\eqSig& w_a \nabla^\top_{\ell} v^a - (\iota^* \odot \nabla \lambda)(\bar{v},\bar{w}) - (\iota^* \ext \lambda)(\bar{w}, \bar{v}) \\
\eqSig& w_a \nabla^\top_{\ell} v^a - K(\bar{v},\bar{w}) \,.
\end{align*}
So, putting this all together, we have that
$$ (\iota_* \bar w)_a \nabla^\top_{\ell} (\iota_* \bar v)^a = \bar{w}_a \bar{\nabla}_{\bar{\ell}} \bar{v}^a\,.$$

\medskip

\subsubsection{Summary of decomposition} \label{subsection:decomposition}
For the purposes of clarity, we fill in table in Equation~\ref{raw-table}:
\begin{widetext}
\begin{equation} \label{full-table}
\begin{aligned}
\lambda_a &\nabla^\top_{\ell} n^a = -\kappa \,, \quad &\nu_a &\nabla^\top_{\ell} n^a = 0\,, \quad &&(\iota_* \bar w)_a \nabla^\top_{\ell} n^a = \bar{w}^a \beta_a\,,\\
\lambda_a &\nabla^\top_{\ell} \ell^a = 0\,, \quad &\nu_a &\nabla^\top_{\ell} \ell^a = \kappa\,, \quad &&(\iota_* \bar w)_a \nabla^\top_{\ell} \ell^a = 0 \,, \\
\lambda_a &\nabla^\top_{\ell} (\iota_* \bar v)^a = 0 \,, \quad &\nu_a &\nabla^\top_{\ell} (\iota_* \bar v)^a = -\bar{v}^a \beta_a \,, \quad &&(\iota_* \bar w)_a \nabla^\top_{\ell} (\iota_* \bar v)^a = \bar{w}_a \bar{\nabla}_{\bar \ell} \bar{v}^a \,, \\
\lambda_a &\nabla^\top_{\iota_* \bar u} n^a = \bar{u}^a \beta_a  \,, \quad &\nu_a &\nabla^\top_{\iota_* \bar u} n^a = 0 \,, \quad &&(\iota_* \bar w)_a \nabla^\top_{\iota_* \bar u} n^a =  \II_{ab} \bar{u}^a \bar{w}^b \,,\\
\lambda_a &\nabla^\top_{\iota_* \bar u} \ell^a = 0\,, \quad &\nu_a &\nabla^\top_{\iota_* \bar u} \ell^a = -\bar{u}^a \beta_a \,, \quad &&(\iota_* \bar w)_a \nabla^\top_{\iota_* \bar u} \ell^a = K_{ab} \bar{u}^a \bar{w}^b \,, \\
\lambda_a &\nabla^\top_{\iota_* \bar u} (\iota_* \bar v)^a = -K_{ab} \bar{u}^a \bar{v}^b \,, &\quad \nu_a &\nabla^\top_{\iota_* \bar u} (\iota_* \bar v)^a = - \II_{ab} \bar{u}^a \bar{v}^b \,, \quad &&(\iota_* \bar w)_a \nabla^\top_{\iota_* \bar u} (\iota_* \bar v)^a = \bar{w}_a \bar{\nabla}_{\bar u} \bar{v}^a\,.
\end{aligned}
\end{equation}
\end{widetext}
Observe that the difference between the spacetime connection along $\mathcal{H}$ and the induced connection on $\mathcal{H}$ is quantified by the extrinsic data specified by $\kappa$, $\beta$, and $\II$, as well as the intrinsic tensor $K$. A subset of these invariants were also noted in~\cite{HopfFrei}.

\subsection{Spacetime-Carrollian relationships: the curvature}
Having established relationships between $\nabla^\top$ and $\bar{\nabla}$, we now have the tools to rewrite the various projections of the spacetime Riemann curvature to $\mathcal{H}$ in terms of extrinsic data (such as $\kappa$, $\beta$, and $\II$) and intrinsic tensors (such as $\bar{R}$ and $K$).

We first consider the tangential component of the spacetime Riemann curvature. Since $T \mathcal{H}$ can be decomposed into $\langle \ell \rangle \oplus H\mathcal{H}$, we will need to consider projections into each component, resulting in (naively) 16 projections to compute. As before, we will assume that $\bar{t},\bar{u},\bar{v},\bar{w} \in \Gamma(H\mathcal{H})$ and their pushforwards are denoted, respectively, by $t,u,v,w \in \Gamma(TM)|_{\mathcal{H}}$.

For arbitrary vector fields $T,U,V,W \in \Gamma(TM)$, we have that
$$T^a U^b V^c W^d R_{abcd} = V_c (\nabla_T \nabla_U W^c - \nabla_U \nabla_T W^c - \nabla_{[T,U]} W^c)\,.$$
Using this identity and table in Equation~\ref{full-table}, we wish to rewrite various projections of spacetime curvatures in terms of the intrinsic curvatures $\bar{R}$ and $K$ and extrinsic curvatures $\beta$ and $\II$. The resulting family of identities is:
\begin{widetext}
\begin{align*}
t^a u^b v^c w^d R_{abcd} &\eqSig \bar t^a \bar u^b \bar v^c \bar w^d (\bar{R}_{abcd} + \II_{ad} K_{bc} - \II_{ac} K_{bd} - \II_{bd} K_{ac} + \II_{bc} K_{ad}) \\
t^a u^b \ell^c w^d R_{abcd} &\eqSig 2 \bar{t}^a \bar{u}^b \bar{w}^d (\bar{\nabla}_{[b} K_{a]d}+ \beta_{[b} K_{a]d}) \\
t^a u^b n^c w^d R_{abcd} &\eqSig 2 \bar{t}^a \bar{u}^b \bar{w}^d (\bar{\nabla}_{[b} \II_{a]d} - \beta_{[b} \II_{a]d}) \\
t^a u^b n^c \ell^d R_{abcd} &\eqSig 2 \bar{t}^a \bar{u}^b (K_{c[a} \II^c_{b]} - \bar{\nabla}_{[a} \beta_{b]}) \\
\ell^a u^b \ell^c w^d R_{abcd} &\eqSig -\bar{u}^b \bar{w}^d (\bar{\nabla}_{\bar{\ell}} K_{bd} + K^2_{bd} - \kappa K_{bd}) \\
\ell^a u^b n^c w^d R_{abcd} &\eqSig \bar{u}^b \bar{w}^d (\bar{\nabla}_b \beta_d - \bar{\nabla}_{\bar{\ell}} \II_{bd} - \beta_b \beta_d - K_{b}^a \II_{ad} -\kappa \II_{bd}) \\
\ell^a u^b n^c \ell^d R_{abcd} &\eqSig -\bar{u}^b (\bar{\nabla}_{\bar{\ell}} \beta_b + \bar{\nabla}_{b} \kappa + 2 \beta_a K_b^a)
\end{align*}
\end{widetext}
The remaining projections can be found by application of the first Bianchi identity. However, note that the first Bianchi is not the traditional statement that $R_{[abc]d} = 0$. Indeed, because the torsion of $\bar{\nabla}$ is non-vanishing, we have a more complicated first Bianchi identity:
$$\bar{R}_{[abc]d} - T_{[ab}{}^e T_{|e|c]d}  - \bar{\nabla}_{[a} T_{bc]d} =0 \,,$$
where $T_{abc} = \bar{g}_{cd} T^d_{ab}$. Given the formula for $T$ in terms of the Ehresmann connection and the intrinsic second fundamental form, we end up with
$$\bar{R}_{[abc]d} + 2 \bar{\nu}_{[a} \bar{\nabla}_{b} K_{c]d} = 0\,.$$
For example, we can use this identity to compute $\bar{\ell}^a \bar{R}_{abcd}$. Note that similar identities can be found for a different, torsion-free (but non-metric) Koszul connection~\cite{mars2013}.

\medskip

To summarize, in this section we have established a list of tensorial quantities that can be used to fully classify a special Carrollian manifold resulting from a null hypersurface embedding. As discussed in the introduction, one of the goals of this article was to construct a family of intrinsic and extrinsic invariants of such embeddings, so that we may better understand the geometric description of black hole thermodynamics. Furthermore, as these geometric invariants can be whittled down to their conformally-covariant pieces, we suspect they will be instrumental in understanding the thermodynamics of black holes that are conformally related to stationary black hole solutions. As noted in~\cite{blitz2}, one can always decompose conformally-covariant geometric quantities along a Riemannian hypersurface in terms of a finitely-generated family of intrinsic and extrinsic geometric objects. We expect that a similar situation will hold for null hypersurfaces, however the finitely-generated family of invariants will have more structure. As such, the invariants produced in this section merely make up these invariants up to first order in derivatives of the metric, and this family will require further development. We leave this as a task for future work.

\section{Examples} \label{sec:examples}

\subsection{Schwarzschild black hole}

As an illustration of our construction, we first consider the simplest (non-trivial) spacetime with an embedded null hypersurface: the Schwarzschild black hole with metric (in Schwarzschild coordinates)
$$ds^2 = (1-\tfrac{2m}{r}) dt^2 - (1-\tfrac{2m}{r})^{-1} dr^2 - r^2 d \Omega^2\,,$$
where $d \Omega^2$ is the metric on the round sphere. With this example, we show that our approach recovers previously known quantities \cite{friedrich1999}. The event horizon is given by the null hypersurface defined by $\mathcal{H} := \{p \in M \;|\; r(p) = 2m\}$.

The geometry distinguishes a fundamental vector field that generates $\mathcal{H}$. Indeed, there is a unique Killing vector field that generates the null hypersurface: $\partial_t$. We thus define $\bar{\ell} := \partial_t|_{\mathcal{H}}$.

As the spacetime is spherically symmetric, there is also a distinguished spatial slice of the hypersurface: the round sphere. So on $\mathcal{H}$, we can foliate by spatial slices
$$S_{t_0} = \{p \in \mathcal{H} \;|\; t(p) = t_0\}\,.$$

Now following the construction of Friedrich, et. al.~\cite{friedrich1999} with $\bar{\ell}$ and $S_{t_0}$, we find that in Gaussian null coordinates, the metric is given by
$$g = \left( 1 - \frac{2m}{r} \right)\ext u^2 + \ext u \ext r - r^2 \ext \Omega^2,$$ 
with canonical rigging vector $n = \partial_r$. In this coordinate system, $\ell = \partial_u$. We also obtain the following coframe
\begin{align*}
    \nu &= \ext u,\\
    \lambda &= \ext r - \left(1-\frac{2m}{r}\right) \ext u, \\
    m^1 &= r \ext \theta,\\
    m^2 &= r \sin\theta \ext\phi. 
\end{align*}

We can now directly compute the invariants of the induced special Carrollian manifold. First, observe that the degenerate metric $\bar{g}$ on $\mathcal{H}$ is given by
$$\bar{g} = -4m^2 \ext \Omega^2\,.$$
Since $\bar{g}$ is independent of $u$,  we have that $K=0$. Furthermore, because  $\bar{\nu}$ is constant under partial differentiation and the torsion vanishes, we find that $\bar{\Gamma}$ are precisely the Christoffel symbols of a round sphere with radius $2m$, and so $\bar{R}$ is that of the same.

In Gaussian null coordinates, we can explicitly compute the extrinsic invariants as well using the table in Equation~\eqref{full-table}:
\begin{align*}
    \kappa &= \frac{1}{4m}\,, \\
    \beta &= 0\,,\\
    \II &=  \frac{1}{4m} \bar{g}\,.
\end{align*}

\subsection{Non-spinning Thakurta metric}

The non-spinning Thakurta metric~\cite{thakurta1981} is a time-dependent conformal rescaling of the Schwarzschild metric:
$$\ext s^2 = e^{-2U(u)}\left( \left( 1 - \frac{2m}{r} \right)\ext u^2 + \ext u \ext r - r^2 \ext \Omega^2  \right),$$

Within the non-spinning Thakurta class of metrics, there are several interesting solutions to general relativity. For example, the Sultana-Dyer solution \cite{sultana2005cosmological} which describes an expanding black hole in an asymptotically Einstein-de Sitter universe is conformally related to the Schwarzschild solution and lies in the Thakurta class of metrics. Among other solutions contained in the non-spinning Thakurta class of metrics are the generalized McVittie solutions with a time dependent mass proportional to the scale factor \cite{maciel2015cosmological,guariento2012realistic} and the more recent black hole solutions or cosmological solutions \cite{mello2017}. Due to the simplicity of this metric we will examine the behaviour of the null hypersurface and compare it with the original Schwarzschild metric.

While the spacetime does not admit a null Killing field that generates the event horizon at $\mathcal{H} = \{p \in M \; | \; r(p) = 2m\}$, it does admit a conformal Killing field which generates the null hypersurface, $\partial_u$, and so we define $\bar{\ell} = \partial_u|_{\mathcal{H}}$. Just as in the Schwarzschild case, it also admits a distinguished spatial foliation of $\mathcal{H}$ parametrized by $u$ and given by the round spheres $S_u$.

In this case, Gaussian null coordinates are not necessary as the canonical rigging vector is chosen to be  $n = e^{2U} \partial_r$ which gives the following coframe 
\begin{align*}
    \nu &= \ext u, \\
    \lambda &= e^{-2U} \left(\ext r - \left(1 - \frac{2m}{r}\right) \ext u \right),\\ 
    m^1 &= e^U r \ext \theta,\\
    m^2 &= e^U r \sin(\theta) \ext \phi. 
\end{align*}

We now explicitly compute the invariants, both intrinsic and extrinsic. Here, the intrinsic second fundamental form is non-vanishing. In fact, this is expected: the conformal rescaling of the Schwarzschild metric, in a way that depends on the $u$ coordinate, implies that the induced metric  on a spatial slice will not preserve lengths when Lie dragged in the $\ell = \partial_u$ direction.

The lowest order intrinsic invariants are then

\begin{align*}
    K_{ab} &= -  U' \bar{g}, \\
    \bar{R}_{ijkl} &= \frac{e^{2U}}{4m^2} (\bar{g}_{ik} \bar{g}_{jl}-\bar{g}_{il} \bar{g}_{jk}). \\
\end{align*}

\noindent Using the expressions in Subsection \ref{subsection:decomposition} the remaining expressions for the curvature tensor can be computed in terms of these. 
 
The extrinsic curvatures can similarly be computed, yielding 
\begin{align*}
    \kappa &= \frac{8U'm^2 +m}{4m^2},\\
    \beta &= 0, \\
    \II &= \frac{e^{2U}}{4m} \bar{g}.
\end{align*}

Comparing the above invariants with that of the Schwarzschild black hole, we see that these null hypersurface embeddings are certainly distinct. We observe that certain quantities (such as the trace-free pieces of $K$ and $\II$, as well as the whole of $\beta$) are conformally covariant, as expected. It will be these quantities (and others, such as the intrinsic Weyl tensor) that will be essential for characterizing the intrinsic and extrinsic invariants of conformal Killing horizons, and can be used to relate the thermodynamics of conformal-to-stationary black holes and their stationary counterparts. 

In addition, from \cite{mello2017} it was noted that these solutions can potentially describe black hole solutions when the null hypersurface is a non-singular surface. This occurs when the conformal factor, $U(t)$ is bounded, and is reflected in the boundedness of the Carrollian invariants.

\subsection{An arbitrary $d-1$ dimensional null hypersurface}

Following from the examples of the horizon of the Schwarzschild solution and the conformal Killing horizon for the non-spinning Thakurta metric, we examine the general case of a null hypersurface in a $d$-dimensional spacetime to concretely illustrate how the Carrollian invariants appear as geometric invariants. 

In the construction outlined in this article, in order to study a given null hypersurace in a Lorentzian spacetime one must choose some spatial slice of the null hypersurface and build coordinates $\{x^A\}$ on this slice. In addition, a representative vector field,  $\bar{\ell}$, of the generating null direction of the null hypersurface must be chosen to be normal to this spatial slice. With these choices, a Gaussian null coordinate system can be constructed locally and the metric of the spacetime takes the form \cite{friedrich1999}:

\beq \ext s^2 = 2 \ext u \ext r + r A \ext u^2 + 2 r B_{A} \ext u \ext x^{A} + C_{AB} \ext x^{\mu}  \ext x^{\nu}\,, \eeq

\noindent where $A, B$ and $C$ are smooth functions of $u,r,x^{A}$ such that $C_{AB}$ is a negative definite $(d-2) \times (d-2)$ matrix. In this coordinate system, the hypersurface, $\mathcal{H}$ is located at $r=0$.

This construction picks out a single null vector-field transverse to the hypersurface, $n = \partial_r$ along with its dual $\nu = g(n,-) = \ext u$ as the principal Ehresmann connection for the Carrollian manifold. In the spacetime, we can construct the coframe:

\beq \begin{aligned} \nu & = \ext u, \\
 \lambda & = \ext r + rA \ext u + rB_{C} \ext x^{C}, \\
\mu^i & = m^i_{~C} \ext x^{C}. \end{aligned} \eeq

\noindent with the dual frame basis, 

\beq \begin{aligned} \ell &= \partial_u - rA \partial_r, \\
n &= \partial_r, \\
m_i &= m_i^{~C} (\partial_{x^C} - r B_{C} \partial_r). \end{aligned} \eeq

\noindent where $m^i_{~C}$ is an invertible matrix with inverse $m_i^{~C}$, satisfying,  $\delta_{ij} m^i_{~A} m^j_{~B} = C_{AB}$. On the Carrollian geometry, the degenerate metric $\bar{g} $ on $\mathcal{H}$ is 

\beq \bar{g} = C_{AB} dx^A dx^B. \nonumber \eeq

To investigate the intrinsic and extrinsic invariants of the induced special Carrollian manifold, we will compute the connection coefficients in the bulk and restrict to the hypersurface. The lowest order intrinsic invariants are then

\beq \begin{aligned}  K_{ij} &= -n(m_{(i|C|}) m_{j)}^{~C}|_{r=0},\\
\bar{R}_{ijkl} &=  2m_{[l}(\Gamma_{|ij|k]})|_{r=0}, \\
\end{aligned} \eeq

\noindent where $$\Gamma_{ijk} = \frac12 (D_{ijk} - D_{jik} - D_{kji}),$$ and $$ D^i_{~jk} =[ m_{[k}(m^i_{~|C|}) m_{j]}^{~C} ].$$ 

\noindent All other components of the curvature tensor and torsion tensor can be reconstructed from these quantities using the expressions in subsection \ref{subsection:decomposition}.

\noindent The extrinsic curvatures are 
\beq \begin{aligned} \kappa &= \frac12 n(rA)_{r=0}, \\
\beta_i &= n (rB_{C}) m_i^{~C}|_{r=0}, \\
\II_{ij} &= \ell(m_{(i|C|}) m_{j)}^{~C} |_{r=0}\,. \end{aligned} \eeq 

While these invariants are quite general, they will be further constrained by asking that the null hypersurface has some physical significance for the spacetime, such as a Killing horizon or a conformal Killing horizon. To do this, it will be necessary to characterize and distinguish the Carrollian geometries associated with such physical null hypersurfaces.

\section{Classification of Carrollian geometries}

We have noted that, under a conformal transformation, a stationary black hole solution may not be mapped to a new black hole solution. For example, in the case of the Schwarzschild solution, a conformal transformation can lead to cosmological solutions \cite{mello2017}. One approach to determining when a solution describes a black hole solution is by locally characterizing, or classifying, the solution. As our methods provide a canonical list of geometric invariants for such a conformally-related spacetime, the invariants so-described should also be usable to determine whether such a spacetime is indeed a black hole solution.

The classification of Lorentzian manifolds can, in principle, be accomplished using the Cartan-Karlhede algorithm \cite{aaman1980}. In this algorithm, an invariantly defined frame is determined by specifying canonical forms of the curvature tensor and its covariant derivatives up to a finite order. The maximum order of covariant differentiation and the uniqueness of this invariant frame is explicitly determined by the algorithm. In addition, in the context of the geometric horizon conjectures, for all weakly isolated horizons, which includes Killing horizons, the classification of black hole solutions using this algorithm identifies a specific curvature invariant that characterizes the horizon \cite{coley2017identification}.

Returning to the problem of determining when a black hole solution containing a Killing horizon yields a new black hole solution under a conformal transformation, this could be achieved locally by characterizing the resulting null hypersurface. Here, we will establish the mathematical framework for such a characterization of null hypersurfaces, with the aim to identify null hypersurfaces acting as black hole horizons in future work.

While null hypersurfaces are embedded in a Lorentzian manifold, we have shown that the intrinsic geometry of a null hypersurface is not Lorentzian but instead is Carrollian and hence admit a torsion tensor in addition to a curvature tensor. Cartan-Karlhede algorithms have been developed for geometries admitting torsion such as teleparallel geometries and Riemann-Cartan geometries \cite{coley2020}. However, these geometries still rely on the Lorentz group as a structure group. This motivates the investigation of a Cartan-Karlhede algorithm for Carrollian geometries equipped with a canonical connection. 

In this section we outline an approach for computing all invariants to locally characeterize a Carrollian manifold using the torsion tensor, the curvature tensor and its covariant derivatives. This approach will rely on the existence of a principal Ehresmann connection in the Carrollian manifold, which is guaranteed in the case of a null hypersurface. Equipped with this preferred Ehsresmann connection, it is possible to determine the components of the connection relative to the frame $\{e_I\} = \{ \ell, m_i\}$ or the coframe $\{\theta^I\} = \{ \nu, \mu^i\}$ as 

\beq \Gamma^I_{~JK} = e^I_{~\nu} e_J^{~\mu} (\partial_\mu e_K^{~\nu} + e_K^{~\lambda} \Gamma^\nu_{~\mu \lambda }) \eeq

\noindent where $e_I = e_I^{~\mu} \partial_\mu$ and $\theta^I = e^J_{~\mu} \ext^\mu$. 

Relative to this frame, we can compute the torsion tensor for the Carrollian manifold:

\beq \bar{T}^i =  \bar{T}^i_{[1j]} \bar{m}_i \otimes \bar{\nu} \wedge \bar{\mu}^j. \eeq  

\noindent In principle, this tensor allows for the fixing of the Carrollian boost parameters. If this is not possible, the curvature tensor may be computed and the components, $\bar{R}_{1ijk}$ can be used to fix the Carrollian boost parameters instead. The remaining $SO(n-1)$ freedom can be used to fix a canonical form for the components $\bar{R}_{ijkl}$. Further fixing of the Carrollian group parameters can be accomplished by computing the respective covariant derivatives of the torsion tensor and the curvature tensor.  We note that it may be advantageous to use Carrollian boosts for which the resulting Ehresmann connection is no longer principal in order to achieve a canonical form for the torsion tensor or curvature tensor.

This suggests the following algorithm. 

\begin{enumerate}
    \item Set the order of differentiation, $q$ to zero.
    \item Compute the derivatives of the torsion tensor and curvature tensor up to order $q$
    \item Determine a canonical form for the $q$-th derivatives of the torsion tensor and curvature tensor. 
    \item Fix the frame parameters using these canonical forms and record the residual frame freedom, denoted as $H_q$.
    \item Find the number of functionally independent components, $t_q$ of the torsion tensor, curvature tensor and their derivatives in the canonical form.
    \item If $\operatorname{dim}(H_q)= \operatorname{dim}(H_{q-1})$ and $t_{q} = t_{q-1}$, set $p+1 =q$ and stop. Otherwise, increment $q$ by $1$ and go to step 2.
\end{enumerate}

\noindent Here, the integer $p$ denotes the highest order of differentiation where new geometric information is introduced.

Given two special Carrollian manifolds, in order to determine their equivalence it is sufficient to determine the canonical form of the torsion tensor, curvature tensors and their respective covariant derivatives for one special Carrolian manifold and attempt to impose the same canonical forms for the other. This is a necessary condition for equivalence, but not sufficient. A sufficient condition follows from comparing the components of their respective canonical forms and solving the resulting equations arising from equating each invariant in the first special Carrollian manifold with the corresponding invariant in the second special Carrollian manifold.

\medskip

\section{Conclusion}




In this paper, we have utilized Carrollian geometric structures to describe, in a canonical way, the intrinsic and extrinsic invariants of special Carrollian manifolds; these can (in many cases) be mapped onto conformal-to-stationary black hole solutions. Key to this construction was the naturality, both from a mathematical and physical perspective, of the induced connection on the Carrollian manifolds. We then we used this method to build conformal invariants that relate the Schwarzschild solution to non-spinning Thakurta metrics in a methodical manner. This straightforward application motivates the use of our approach to compute the conformal invariants for more general metrics and study their thermodynamical properties. 

Furthermore, as the developed theory introduces a complete first-order family of geometric tensorial invariants for these solutions, in the future we hope to be able to use these invariants to determine whether a spacetime truly does contain a black hole solution, rather than a cosmological solution. The question of when a dynamical black hole is conformal to a stationary solution is significant in the context of the geometric horizon conjectures \cite{coley2017geometric}. In comparison with  the standard curvature invariants used to detect the horizon are compared to the class of conformal invariants constructed in \cite{mcnutt2017scalar, mcnutt2017curvature}, the level sets constructed from our set of invariants may yield different hypersurfaces and this suggests several possible horizon candidates where it is not obvious which hypersurface will give the appropriate black hole boundary. 


From a purely mathematical perspective, as a result of developing this formalism, we have provided clarity on the types of Carrollian geometries that can be associated to null hypersurface embeddings. 
Indeed, outside of applications to physical black holes solutions, 
we expect that interesting and fruitful classification results of more general Carrollian geometries will follow naturally from the framework developed here. Beyond just the geometries investigated in this article, elsewhere there has also been some work on conformal Carrollian structures~\cite{Duval2014,Ciambelli2019,Herfray2020} which may dove-tail nicely with our constructions. In future work we hope to tie in our constructions with those structures as well.

\vspace{5 mm }
 \begin{acknowledgments}
SB was supported by a COST Action CaLISTA grant (reference number E-COST-GRANT-CA21109-1cd88727). Further support was given by the Czech
Science Foundation (GACR) grant GA22-00091S and the Operational Programme Research Development and Education Project No. CZ.02.01.01/00/22-010/0007541. SB would also like to acknowledge the UiT Department of Mathematics and Statistics for their gracious hospitality. DM is supported by the Norwegian Financial Mechanism 2014-2021 (project registration number 2019/34/H/ST1/00636).
 \end{acknowledgments}

\bibliographystyle{unsrt}
\bibliography{NHCSbib}


\end{document}